\documentclass[10pt,onecolumn,draftclsnofoot,twoside]{IEEEtran}
\usepackage{amsmath,amsfonts}
\usepackage{algorithmic}
\usepackage{algorithm}
\usepackage{amssymb}
\usepackage{subfigure}
\usepackage{color}

\usepackage{float} 

\usepackage{array}
\usepackage[caption=false,font=normalsize,labelfont=sf,textfont=sf]{subfig}
\usepackage{textcomp}
\usepackage{stfloats}
\usepackage{url}
\usepackage{bm}
\usepackage{setspace}
\usepackage{verbatim}
\usepackage{multirow} 
\usepackage{ntheorem}
\usepackage{graphicx}
\usepackage{cite}
\usepackage{amsmath}

\allowdisplaybreaks[4]

\theorembodyfont{\upshape}
\theoremheaderfont{\it}
\qedsymbol{\square}
\theoremseparator{:}
\newtheorem{theorem}{\indent Theorem}

\newtheorem*{proof}{\indent Proof}
\newtheorem{remark}{\indent Remark}

\newtheorem{proposition}{\indent Proposition}

\makeatletter

\newcommand{\Rmnum}[1]{\expandafter\@slowromancap\romannumeral #1@}
\makeatother
\begin{document}

\title{\huge Digitalizing Over-the-Air Computation via The Novel Complement Coded Modulation}

\author{Zhixu Wang, Jiacheng Yao, \IEEEmembership{Graduate Student Member, IEEE}, 
        Wei~Xu,~\IEEEmembership{Fellow, IEEE},\\
        Wei~Shi,~\IEEEmembership{Member, IEEE},
        and Kaibin~Huang,~\IEEEmembership{Fellow, IEEE}
\vspace{-0.35cm}
\thanks{Z. Wang, J. Yao, and W. Xu are with the National Mobile Communications Research Laboratory, Southeast University, Nanjing 210096, China (e-mail: \{zhixuwang, jcyao, wxu\}@seu.edu.cn). J. Yao, W. Xu, and W. Shi are with the Purple Mountain Laboratories, Nanjing 211111, China (e-mail: shiwei1@pmlabs.com.cn). K. Huang is with the Department of Electrical and Electronic Engineering, The University of Hong Kong, Hong Kong SAR, China (e-mail: huangkb@hku.hk).}
\vspace{-0.35cm}
}



\maketitle

\begin{abstract}
To overcome inherent limitations of analog signals in over-the-air computation (AirComp), this letter proposes a two’s complement-based coding scheme for the AirComp implementation with compatible digital modulations. Specifically, quantized discrete values are encoded into binary sequences using the two’s complement and transmitted over multiple subcarriers. At the receiver, we design a decoder that constructs a functional mapping between the superimposed digital modulation signals and the target of computational results, theoretically ensuring asymptotic error-free computation with the minimal codeword length. To further mitigate the adverse effects of channel fading, we adopt a truncated inversion strategy for pre-processing. Benefiting from the unified symbol distribution after the proposed encoding, we derive the optimal linear minimum mean squared error (LMMSE) detector in closed form and propose a low-complexity algorithm seeking for the optimal truncation selection. Furthermore, the inherent importance differences among the coded outputs motivate an uneven power allocation strategy across subcarriers to improve computational accuracy. Numerical results validate the superiority of the proposed scheme over existing digital AirComp approaches, especially at low signal-to-noise ratio (SNR) regimes.
\end{abstract}


\vspace{-0.4cm}
\section{Introduction}
\IEEEPARstart{T}{he} next-generation wireless communication network is expected to deeply integrate computational capabilities to support emerging intelligent applications \cite{ref1,zhyang}. 
Along this line, task-oriented integrated sensing, communication, and computation (ISCC) has been recognized as a promising paradigm for low-latency edge inference, where energy-efficient designs via joint optimization across communication and computation modules have been actively studied \cite{Yao2025ISCC}. 
In such collaborative edge-intelligent systems, fast and scalable aggregation over massive devices is often required. 
Over-the-air computation (AirComp) has attracted considerable attention as it unifies communication and computation. Specifically, it exploits the superposition property of multiple access channels (MACs) to achieve simultaneous aggregation of uncoded analog signals, enabling automatic computation over the air \cite{gxzhu,survey}. By reusing the limited spectrum among multiple devices, AirComp improves spectral efficiency and reduces computational latency \cite{ref2}. Despite these advantages, its analog nature presents several challenges. Firstly, reliance on analog processing leads to compatibility issues with modern digital communication systems and incurring additional hardware costs \cite{ref3}. Secondly, analog signals are highly susceptible to channel fading, interference, and noise, reducing computational robustness \cite{ref4}. Hence, integrating digital signals into AirComp is crucial to enhance both compatibility to commercial communication systems and computational robustness.

The main challenge in digital AirComp is establishing a unique mapping between the superimposed digital signals and the target computational outcomes. One approach is to partition the received signal space and assign each partition to a corresponding result. For example, the authors of \cite{ref5}~proposed a binary phase shift keying (BPSK)-based one-bit digital aggregation scheme for majority voting, where the voting result was determined by the sign of the received signal. This was later extended in \cite{ref6} using FSK modulation to improve voting accuracy. For more general functions, the authors of \cite{ref7} proposed ChannelComp, which optimizes the digital modulation constellation to establish the unique mapping. However, as the number of devices or the modulation order increases, the number of possible outcomes grows exponentially, resulting in a severe increase in detection complexity at the receiver.

To avoid the complexity boosting at receivers, an alternative is to employ relatively simple digital modulations while leveraging coding technique to establish a functional relationship between the superimposed signals and the computational results. For example, a modulation on conjugate-reciprocal zeros (MOCZ)-based digital AirComp method was proposed in \cite{ref8} for majority voting. Moreover, the author in \cite{ref10} used the balanced number system and designed a corresponding encoder and decoder for the digital AirComp. Although easy to deploy, the coding-based approach maps a single continuous value into a sequence of symbols with excessive codeword length, thereby reducing the computational efficiency of AirComp.

In computer systems, signed numbers are usually represented using the two's complement, enabling unified arithmetic operations via addition \cite{code}, which naturally aligns with the superposition over MACs.  Motivated by this, we adopt the two's complement for coding, thereby achieving error-free computation with the provable minimal codeword length. Although  binary representation based coding has been discussed in \cite{binary}, its extension and dedicated design for signed numbers remain unexplored. Moreover, coding alone is insufficient to combat dynamic wireless environments. To enhance robustness, transceiver optimization and resource allocation remain essential yet challenging. In particular, maximum likelihood (ML) detection is no longer optimal from a computational perspective, necessitating alternative detectors. To this end, we highlight two key properties of the proposed coding scheme: 1) unifying data sources into a standard Bernoulli distribution, which facilitates the design of optimal truncation and linear minimum mean squared error (LMMSE) detector;  2) introducing importance differences among bits, which enables an uneven power allocation to enhance computational accuracy under stringent resource constraints. Finally, the scheme is extended to multiple-input multiple-output (MIMO) settings.

\begin{figure*}[!t]
  \centering
  \centerline{\includegraphics[width=5.8in]{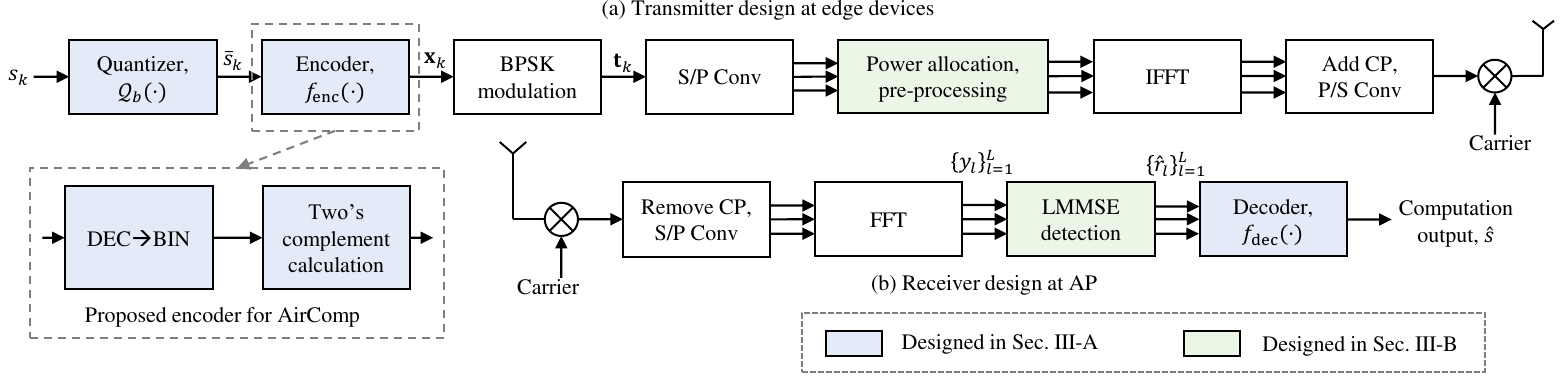}}
  \caption{Transceiver design of the proposed complement coded digital AirComp.}\label{fig1}
  \vspace{-0.4cm}
\end{figure*}

\vspace{-0.3cm}
\section{System Model}
We consider a single-input single-output (SISO) AirComp system comprising an access point (AP) and $K$ edge devices. Let $s_k\in \mathbb{R}$ denote the data generated by device $k$, which is drawn from a given distribution. The AP aims to compute a target function of $\{s_k\}_{k=1}^K$. Without loss of generality, we focus on the sum $s=\sum_{k=1}^K s_k$. To achieve this, we employ digital AirComp, whose common process is illustrated in Fig.~\ref{fig1} and outlined as follows.

At the transmitter, to enable digital AirComp, each continuous value $s_k$ is first quantized using a $b$-bit quantizer $\mathcal{Q}_b(\cdot)$, yielding $\bar{s}_k=\mathcal{Q}_b(s_k)$. The quantized value $\bar{s}_k$ is further encoded into a sequence with $L$ binary elements, and then aggregated over the MACs after digital modulation. Let $f_{\text{enc}}(\cdot)$ denote the encoder with codeword length $L$, i.e., 
\begin{align}
    f_{\text{enc}}(\bar{s}_k)\!=\!\mathbf{x}_k\!=\!(x_{k,1},x_{k,2},\cdots,x_{k,L}),\, x_{k,l}\in \{0,1\}.
\end{align}
Then, we adopt BPSK for digital modulation and hence each coded bit $x_{k,l}$ is mapped to the symbol $t_{k,l}=2x_{k,l}-1$. For the simultaneous transmission of $L$ symbols, orthogonal frequency division multiplexing (OFDM) modulation is adopted to divide the available spectrum into $L$ orthogonal subchannels\footnote{Standard OFDM techniques can be employed to achieve synchronization across multiple devices.}. All devices transmit their $L$ symbols in parallel over these~subchannels. The received signal on the $l$-th subcarrier is given~by
\begin{align}\label{eq3}
    y_l=\sum_{k=1}^{K} h_{k,l}\rho_{k,l}t_{k,l}+n_l,
\end{align}
where $h_{k,l}$ denotes the channel between device $k$ and the AP, $\rho_{k,l}$ denotes the pre-processing factor of device $k$ on subcarrier $l$, and $n_l\sim \mathcal{CN}(0,\sigma_{l}^{2})$ is the additive Gaussian noise with power $\sigma_l^2$. For simplicity, large-scale fading is considered being compensated by power control.

At the receiver, the AP first detects the desired superposed outcome on each subcarrier, i.e., $r_l=\sum_{k=1}^K x_{k,l}$, from $y_l$, yielding the estimate $\hat{r}_l$. Based on $\{\hat{r}_l\}_{l=1}^L$, we estimate $s$ by
\begin{align}
    \hat{s}=f_{\text{dec}}(\hat{r}_1,\hat{r}_2,\cdots,\hat{r}_L),
\end{align}
where $f_{\text{dec}}(\cdot)$ is a well-designed decoder satisfying 
\begin{align}\label{constraint}
    f_{\text{dec}}\left(\sum_{k=1}^K f_{\text{enc}}(\bar{s}_k)\right) =\sum_{k=1}^K \bar{s}_k.
\end{align}
Constraint (\ref{constraint}) ensures that, the receiver can recover the sum of the quantized values without error in an ideal case.

We note that the challenge for implementing the coded digital AirComp lies on the joint design of the encoder, $f_{\text{enc}}(\cdot)$, and decoder, $f_{\text{dec}}(\cdot)$, to meet the constraint in (\ref{constraint}), while minimizing the codeword length $L$ to improve spectral efficiency. Moreover, to enhance computational accuracy under channels fadings, further optimization of transceiver and resource allocation is required, though all remain highly coupled. We focus on addressing the above issues in the next section.

\section{Proposed Complement Coded Digital AirComp}
In this section, we first present the coding design for~digital AirComp and then optimize the pre-processing, power allocation, and detection to mitigate the effects of channel fading and noise. The corresponding modules are illustrated in Fig.~\ref{fig1}.

\subsection{Encoding and Decoding Design}
Inspired by the properties of the two's complement representation, we construct the encoder and decoder as follows:
\begin{itemize}
\item Encoder: To guarantee that the quantized discrete values fall within the range of the two’s complement representation, we first specify the following quantizer
\begin{align}\label{quantizer}
\bar{s}_k = \mathcal{Q}_b(s_k) \triangleq \frac{1}{\zeta} \lfloor \zeta s_k \rfloor,
\end{align}
where $\zeta$ is set as $\frac{2^{b-1}}{\max\{s_k\} + \varepsilon}$ and $\varepsilon$ is a small positive constant satisfying $\varepsilon \ll 2^{1-b}$. After scaling $\bar{s}_k$ by $\zeta$, we obtain an integer in the set $\{-2^{b-1}, -2^{b-1}\!+\!1, \cdots, 0, 1, \cdots, 2^{b-1}\!-\!1\}$. The encoder converts the signed decimal integer $\zeta \bar{s}_k$ into its two's complement sequence $\{x_{k,b}, \ldots, x_{k,1}\}$, which is the encoded output.
\item Decoder: The decoder $f_{\text{dec}}(\cdot)$ is defined as
\begin{align}\label{encoder}
    f_{\text{dec}}(r_1,r_2,\!\cdots\!,r_L)\!=\!\frac{1}{\zeta}\!\left( \sum_{l=1}^{L-1} r_l 2^{l-1}\!-\!r_L 2^{L-1}\right).
\end{align}
\end{itemize}

\begin{theorem} \label{theo1}
The proposed encoder and decoder satisfy the constraint in (\ref{constraint}) while achieving the minimal codeword length.
\end{theorem}
\begin{proof}
    Please refer to Appendix \ref{appa}. \hfill $\square$
\end{proof}

\begin{remark}
    \emph{Theorem \ref{theo1}} shows that the proposed encoder and decoder establish a deterministic mapping using the minimal codeword length $L=b$. This substantially reduces resource overhead compared with \cite{ref10}, where the codeword length is much larger, i.e., $L\geq (\beta-1)\log_\beta2^b$ for any odd $\beta\geq 3$. Although binary representation based scheme \cite{binary} has the~same length, the proposed method natively supports signed number computations, making it more suitable for practical tasks.
\end{remark}

Besides, the proposed coding scheme exhibits two~advantageous properties not explicitly reported in existing designs. They provide a solid foundation for further optimizations.

\emph{Property 1}: The output coded signal $x_{k,l}$ follows a standard Bernoulli distribution under the assumption that the source $s_k$ is symmetrically distributed. Hence, the proposed encoder transforms arbitrary data sources into a standard Bernoulli distribution, removing the need for prior knowledge of the source distribution at the receiver.

\emph{Property 2}: As observed in (\ref{encoder}), the aggregation results from different subcarriers, $r_l$, are weighted differently, resulting in varying importance. This motivates the implementation of uneven power allocation. In extreme cases, lower-priority computations can be disregarded to prioritize the accuracy of higher-priority ones, thereby improving the robustness.

\subsection{Pre-processing, Power Allocation, and Detection Design}
Although error-free computation is theoretically achievable, practical over-the-air aggregation suffers from channel fading and additive noise, necessitating further optimization of transceiver and resource allocation. Due to the difficulty in directly deriving explicit objectives and the coupling among optimization variables, solving the resulting complicated problem becomes intractable. To this end, we adopt a decomposition strategy that enables individual yet effective optimization.

Specifically, we first adopt a typical truncated channel inversion scheme \cite{gxzhu}, which shares a similar form with the optimal one. The pre-processing factor for device $k$ on subcarrier $l$ is
\begin{equation}\label{eq7}
    \rho_{k,l} = 
    \begin{cases}
       \frac{{\sqrt{p_l}} h_{k,l}^{*}}{|h_{k,l}|^2},  & k\in \mathcal{K}_l, \\
       0, & k\notin \mathcal{K}_l,
    \end{cases}
\end{equation}
where $\mathcal{K}_l$ is the set of activated devices on the $l$-th subcarrier, and $p_l$ is a scaling factor for ensuring power constraint. 
Then, the scaling factor $p_l$ satisfies
\begin{align}\label{scaling}
    \frac{p_l}{|h_{k,l}|^2} \mathbb{I}_{k,l}\leq P_{k,l},\, \forall k,l,
\end{align}
where $\mathbb{I}_{k,l}=1$ only if $k\in \mathcal{K}_l$ and otherwise $\mathbb{I}_{k,l}=0$, $P_{k,l}$ denotes the allocated transmit power budget of device $k$ on subcarrier $l$. Besides, $\sum_{l=1}^L P_{k,l}=P_{k,\max}$ and $P_{k,\max}$ denotes the total transmit power budget of device $k$.  Considering \emph{Property 2}, uniform power allocation across subcarriers inevitably results in performance degradation. Inspired by the mathematical formalism in (\ref{encoder}), we propose a heuristic strategy that prioritizes high-criticality computations\footnote{We note that the coded output bits are not strictly independent; instead, they exhibit source-dependent correlations. These unknown and distribution-dependent correlations render the power allocation across subcarriers highly coupled and difficult to optimize in a explicit way. To this end, we adopt a heuristic strategy that decouples the design across subcarriers, thereby enhancing applicability under diverse and unknown source distributions.}. Specifically, the power allocation follows $\frac{P_{k,l}}{P_{k,l-1}}=\varpi,\,\forall l=2,\cdots,L$, where $\varpi>1$ is a predetermined constant and $P_{k,1}$ is chosen to satisfy $P_{k,1}\leq \frac{1 - \varpi}{1-\varpi^{b}} P_{k,\max}$ thereby ensuring the maximum power constraint. Now, the allocated power follows a geometric sequence, consistent with the weight growth of $\{r_l\}_{l=1}^L$ in (\ref{encoder}). Moreover, the transmit power of truncated devices is set to zero, which inevitably causes power wastage. It can be alleviated by performing power allocation again after the set of active devices is determined.

With the given pre-processing strategy, we are available to derive the LMMSE detector in the following proposition.
\begin{proposition} \label{pro1}
    The LMMSE detector on subcarrier~$l$~is
    \begin{align}\label{eq8}
        \hat{r}_l=\lambda_l \Re\{y_l\} + \mu_l,
    \end{align}
    where $\lambda_l=  \frac{\sqrt{p_l} |\mathcal{K}_l|}{2p_l |\mathcal{K}_l| + \sigma_l^2}$ and $\mu_l= \frac{K}{2}$.
\end{proposition}
\begin{proof}
    Please refer to Appendix \ref{appb}. \hfill $\square$
\end{proof}

\begin{remark}
    With low  SNR, i.e., $\frac{p_l}{\sigma_l^2} \!\to \!0$, we have $\lambda_l =0$, implying that the receiver estimates the computation result using the mean value and thereby reducing excessive computational errors. Hence, unlike pure communication systems, the LMMSE detector is optimal from a computational perspective.
\end{remark}

Now, with the LMMSE receiver in (\ref{eq8}), the minimum achievable MSE for $r_l$ is expressed as follows
\begin{align} \label{eq13}
    e_l= \frac{2p_l |\mathcal{K}_l|\left(K-|\mathcal{K}_l|\right)+K\sigma_l^2}{8p_l|\mathcal{K}_l|+4\sigma_l^2}.
\end{align}
Note that $e_l$ depends on the scaling factor and the truncation selection, which motivates us to further optimize them for minimizing MSE.  
The corresponding optimization problem is\footnote{{By solving problem (\ref{problem}), the AP feeds the optimal $\mathcal{K}_l$ back to each device. The additional overhead is a $L$-bit vector and thus negligible in practice.}}
\begin{align}\label{problem}
    \mathop{\text{minimize}}_{p_l,\mathcal{K}_l} \enspace &e_l\nonumber \\
    \text{subject to}\enspace &\frac{p_l}{|h_{k,l}|^2}\!\leq \!P_{k,l},\, \forall k\!\in \!\mathcal{K}_l,\,\mathcal{K}_l \!\subset\! \{1,2,\cdots,K\}.
\end{align}
We notice that $e_l$ monotonically decreases as $p_l$ increases. Hence, with any given $\mathcal{K}_l$, the optimal scaling factor equals
\begin{align}\label{optp}
    p_l^{\text{opt}}=\min_{k\in \mathcal{K}_i}\{ |h_{k,l}|^2 P_{k,l}\}.
\end{align}
Then, to optimize $\mathcal{K}_l$, we propose Algorithm \ref{alg0}, whose computational complexity is $\mathcal{O}(K\mathrm{log}K)$ due to the sorting operation. Regarding the optimality, we derive the following theorem.
\begin{theorem}\label{theo2}
    The solution obtained by Algorithm 1 is the optimal solution to problem (\ref{problem}).
\end{theorem}
\begin{proof}
    Please refer to Appendix \ref{appc}. \hfill $\square$
\end{proof}

\begin{algorithm}[!t]
\caption{Greedy Algorithm for Solving (\ref{problem})} \label{alg0}
\renewcommand{\algorithmicrequire}{\textbf{Input:}}
\renewcommand{\algorithmicensure}{\textbf{Ouput:}}
\begin{algorithmic}[1]  
\REQUIRE $\{h_{k,l}\}_{k=1}^K$ and $\sigma_l^2$ $\enspace$ \textbf{Output:} $\mathcal{K}_l$
\STATE Sort the $K$ devices by channel gain to satisfy $|h_{k_1,l}|\geq |h_{k_2,l}|\geq\cdots\geq |h_{k_K,l}|$.
\STATE Initialize $\text{tmp}=+\inf$ and $\mathcal{C}_l=\emptyset$.
\FOR{$n=1:K$}
\STATE Update  $\mathcal{C}_l=\mathcal{C}_l\cup{k_n}$.
\STATE Calculate $p_l=\vert h_{k_n}\vert^2P_{k,l}$ and its corresponding $e_l$.
\IF{$e_l<\text{tmp}$}
\STATE Set $\mathcal{K}_l=\mathcal{C}_l$ and $\text{tmp}=e_l$.
\ENDIF
\ENDFOR
\end{algorithmic}
\end{algorithm}


\begin{remark}
    We note that both the latency and detection complexity of the proposed scheme do not scale with the number of devices, $K$. Moreover, the complexity of Algorithm~\ref{alg0} grows only mildly with $K$. Hence, the proposed scheme is naturally compatible with massive access devices  \cite{MD}.
\end{remark}
\vspace{-0.3cm}
\subsection{Extension to MIMO Scenarios}
 We further extend to a MIMO system, where each device has $N_t$ transmit antennas and the AP has $N_r$ receive antennas. The $k$-th device employs a normalized beamformer $\mathbf{f}_{k,l}\!\in\!\mathbb{C}^{N_t\times1}$, while the AP adopts a linear receive beamformer  $\mathbf{w}\!\in\!\mathbb{C}^{N_r\times1}$ 
on subcarrier $l$. The received signal is rewritten as
\begin{align}
y_l
=\sum_{k=1}^{K} \mathbf{w}^{H}\mathbf{H}_{k,l}\mathbf{f}_{k,l}\rho_{k,l}t_{k,l}
+\mathbf{w}^{H} \mathbf{n}_l,
\end{align}
where $\mathbf{H}_{k,l}$ denotes the  channel between the AP and the $k$-th device on subcarrier $l$.  After beamforming, the MIMO link can be equivalently represented by  effective scalar channels  $h_{k,l}\triangleq \mathbf{w}^{\mathit{H}}\mathbf{H}_{k,l}\mathbf{f}_{k,l}$, which preserves the SISO structure and thus allows all proposed operations to remain directly applicable. Besides, existing beamforming optimization methods \cite{optbf} can be employed to fully exploit array gain.
\vspace{-0.3cm}
\section{Numerical Results}

In this section, we present simulation results. We assume the same transmit power budget $P_{\max}$ of each device and the same noise power across subcarriers, i.e., $\sigma_l^2=\sigma^2,\forall l$. The SNR is defined as $\frac{P_{\max}}{\sigma^2L}$. Moreover, a multi-path Rayleigh fading model is adopted \cite{ofdmchannel}, where each subcarrier experiences independent small-scale fading. The channel of device $k$ on subcarrier $l$ is modeled as $h_{k,l} = \sum_{m=1}^{M} h_{l,m}^k \, e^{j \frac{2\pi \tau_m^k l}{L}}$,
where $h_{l,m}^k$ is the complex gain of the $m$-th path, $\tau_m^k$ is its relative delay, and $M$ is the number of multi-path taps. Unless otherwise specified, the other parameters are set as follows: the number of users, $K = 20$, the number of quantization bits, $b = 8$, and the number of subcarriers, $L=8$. For performance comparison, we consider the following baselines. 
\begin{itemize}
    \item \textbf{Analog} \cite{gxzhu}: The typical analog scheme adopts truncated channel inversion with fixed threshold ($\gamma$), and it repeats the computation on different subcarriers for averaging.
    \item \textbf{Binary+ML} \cite{binary}: Using the binary representation for coding, together with BPSK modulation, truncated channel inversion, and ML detection.
    \item \textbf{Balance} \cite{ref10}: Using balanced number system for coding.
    \item \textbf{Bit-slicing} \cite{bitslicing}: Using bit slicing  for coding and employing per-subcarrier maximum a posterior (MAP) detection.
\end{itemize}

\begin{figure}[!t]
    \centering
    \subfigure[Uniformly distributed source]{ \includegraphics[width=0.93\linewidth]{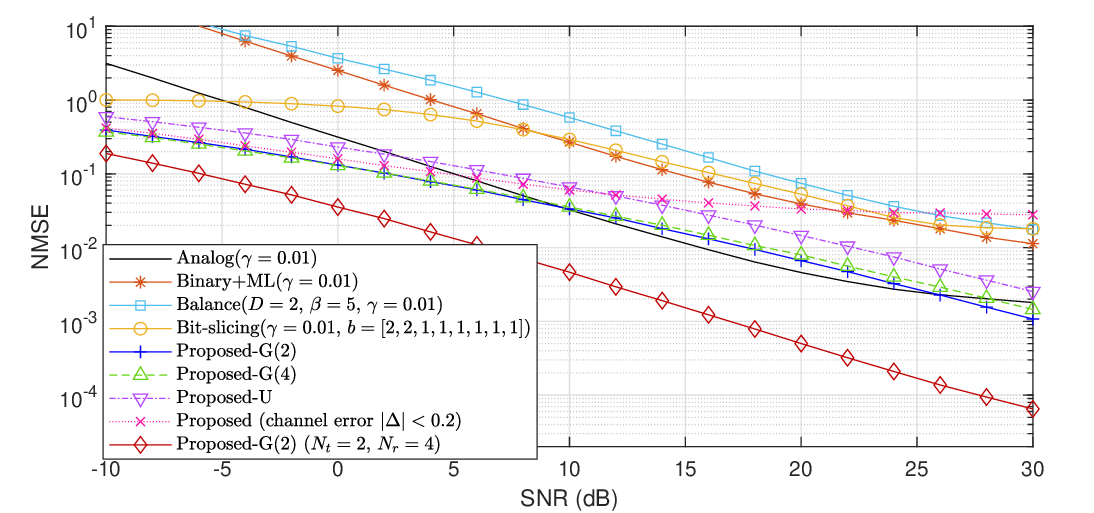}}
    \subfigure[Gaussian source]{	\includegraphics[width=0.93\linewidth]{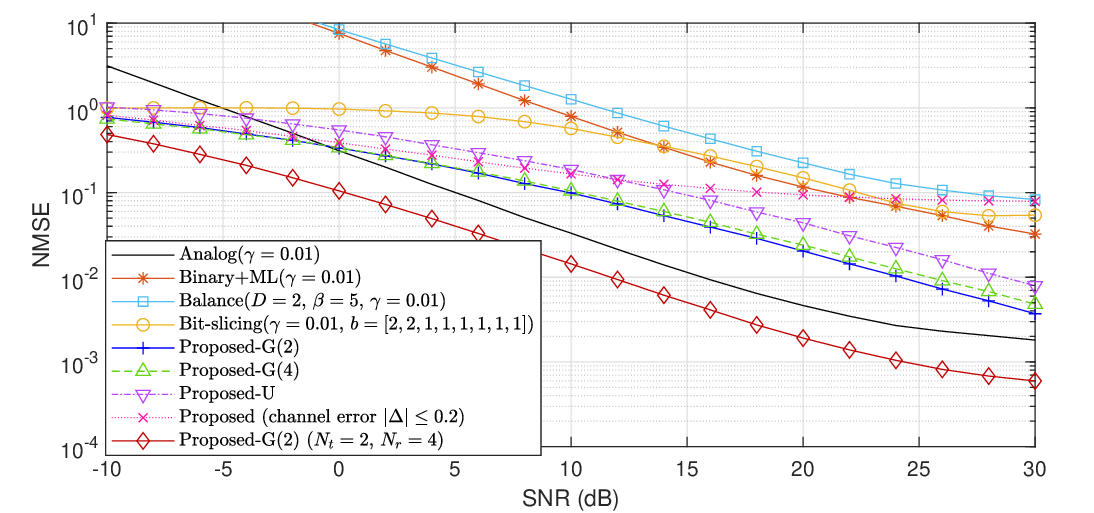}}
    \caption{\centering {NMSE versus SNR levels.}} \label{fig2}
    \vspace{-0.5cm}
\end{figure}

In Fig. \ref{fig2}, we evaluate the normalized MSE (NMSE) versus the SNR levels with uniformly/Gaussian distributed source data. The legend ``Proposed-U" represents the proposed scheme with uniform power allocation while ``Proposed-G" denotes the proposed scheme employing the geometric power allocation strategy. At low SNRs, the proposed scheme achieves clear gains over analog baselines. In some cases, however, it becomes slightly inferior due to quantization errors. This effect is particularly pronounced for Gaussian sources, whose unbounded nature results in significant quantization errors. Compared with digital baselines, the proposed scheme consistently performs better. Since the conventional binary-representation and Bit-slicing schemes offer consistent coding efficiency, the performance gain mainly arises from the subsequent transceiver and resource allocation design. Compared with the Balance scheme, the proposed approach attains higher coding efficiency and thus yields significant improvements.

Moreover, the proposed power allocation strategy further boosts performance by allocating more resources to critical subcarriers, with a larger $\varpi$ required for more unbalanced allocation under low SNR. {Besides, we evaluate performance in MIMO settings, confirming that the proposed scheme extends naturally to such scenarios and that the array gain provided by multiple antennas further improves computational accuracy.}
{Finally, we evaluate the performance under channel estimation errors ($|\Delta|<0.2$) and observe that such errors have negligible impact at low SNRs, demonstrating the robustness of the proposed method in low-SNR regimes.}

\section{Conclusion}
In this paper, we proposed a novel complement coded digital AirComp scheme that exploited the property of two’s complement to achieve error-free computation. Moreover, we optimized pre-processing, power allocation and detection by using the unified distribution of the coded outputs and the intrinsic importance differences among bits. Furthermore, {developing nonlinear detectors such as message-passing for reducing MSE, as well as advanced channel coding methods for error correction, are promising directions in the future.}

\appendices
\section{Proof of Theorem \ref{theo1}}\label{appa}
According to \cite[Eq. (1)]{code}, the decimal integer $\zeta \bar{s}_k$ and its two's complement representation $\{x_{k,b}, \ldots, x_{k,1}\}$ satisfies $\zeta \bar{s}_k = \sum_{l=1}^{L-1}x_{k,l}2^{l-1}-x_{k,L}2^{L-1}$. Then, we verify that 
    \begin{align}
        \sum_{k=1}^K \bar{s}_k &=\frac{1}{\zeta} \sum_{k=1}^K \left(\sum_{l=1}^{L-1}x_{k,l}2^{l-1}-x_{k,L}2^{L-1}\right)\nonumber\\
        &=\frac{1}{\zeta}\left( \sum_{l=1}^{L-1}\left(\sum_{k=1}^K x_{k,l}\right)2^{l-1}-\left(\sum_{k=1}^K x_{k,L}\right)2^{L-1}\right)\nonumber \\
        &=f_{\text{dec}}(r_L,\!\cdots\!,r_2,r_1).
    \end{align}
{We next show that any scheme with code length $L<b$ cannot satisfy constraint (\ref{constraint}). Since $\bar{s}_k$ takes $2^b$ values while $\mathbf{x}_k$ has $2^L$ codewords, $L<b$ implies that the number of codewords is smaller than the number of source values. Consequently, at least two distinct $\bar{s}_k$ values are mapped to the same codeword, making it impossible for the receiver to distinguish them and achieve error-free computation. Hence, the minimum code length is $L=b$, which is attained by the proposed scheme.}
    
\section{Proof of Proposition \ref{pro1}}\label{appb}
Utilizing \emph{Property 1}, we obtain that $r_l =\sum_{k=1}^K x_{k,l}$ follows Bernoulli distribution $\mathcal{B}(K,0.5)$. Moreover, with the pre-processing in (\ref{eq7}), we express the received signal $y_l$ as 
\begin{align}
    y_l =2\sqrt{p_l} \sum_{k \in \mathcal{K}_l} x_{k,l} - \sqrt{p_l} |\mathcal{K}_l| + n_l,
\end{align}
and we define $\tilde{r}_l\triangleq \sum_{k\in \mathcal{K}_l} x_{k,l} \sim  \mathcal{B}(\vert \mathcal{K}_l\vert,0.5)$. With the detector in (\ref{eq8}), we calculate the MSE on subcarrier $l$ as 
\begin{align} \label{ae1}
e_l \!\!&=\!\! \mathbb{E}[(r_l\!\! -\! \hat{r}_l)^2]\!\!=\! \!\mathbb{E}\!\!\left[\!\left(r_l \!-\!\! 2\!\sqrt{p_l} \lambda_l \tilde{r}_l \!+\!\! \sqrt{p_l} \lambda_l |\mathcal{K}_l| \!-\! \!\lambda_l \Re\{n_l\}\! -\!\! \mu_l\!\right)^2\!\right]\nonumber \\
&\overset{\text{(a)}}{=} \mathbb{E}[r_l^2] 
- 4\sqrt{p_l} \lambda_l \mathbb{E}[r_l \tilde{r}_l] 
+ 2\sqrt{p_l} \lambda_l |\mathcal{K}_l| \mathbb{E}[r_l] 
- 2\mu_l \mathbb{E}[r_l] \notag \\
&\quad + 4p_l \lambda_l^2 \mathbb{E}[\tilde{r}_l^2] 
- 4p_l \lambda_l^2 |\mathcal{K}_l| \mathbb{E}[\tilde{r}_l] 
+ 4\sqrt{p_l} \lambda_l \mu_l \mathbb{E}[\tilde{r}_l] \notag \\
&\quad + p_l \lambda_l^2 |\mathcal{K}_l|^2 
- 2\sqrt{p_l} \lambda_l |\mathcal{K}_l| \mu_l 
+ \frac{1}{2} \sigma_l^2 \lambda_l^2 + \mu_l^2,
\end{align}
where the equality in (a) comes from the fact that $\Re\{n_l\}\sim \mathcal{N}\left(0,\frac{\sigma_l^2}{2}\right)$. Then, the expectations in (\ref{ae1}) are calculated as
\begin{align} \label{ae2}
    &\mathbb{E}\left[r_l \right]=\frac{K}{2}, \enspace \mathbb{E}\left[r_l^2 \right]=\frac{K^2+K}{4},\mathbb{E}\left[\tilde{r}_l \right]=\frac{|\mathcal{K}_l|}{2},\nonumber \\
    &\mathbb{E}\left[\tilde{r}_l^2 \right]=\frac{|\mathcal{K}_l|^2+|\mathcal{K}_l|}{4},\mathbb{E}\left[r_l \tilde{r}_l\right]\!=\!\frac{|\mathcal{K}_l|(K+1)}{4}.
\end{align}
Substituting (\ref{ae2}) into (\ref{ae1}), we have  
\begin{align}
    e_l\!=\!&\left(\!p_l |\mathcal{K}_l| \!+\! \frac{1}{2} \sigma_l^2\!\right) \lambda_l^2 \!
    \!-\! \!\sqrt{p_l} |\mathcal{K}_l| \lambda_l 
    \!-\!\! K \mu_l \!+\! \mu_l^2 \!+\! \!\frac{K^2 \!\!\!+\! K}{4}.
\end{align}
Hence, by setting the derivatives w.r.t. $\lambda_l$ and $\mu_l$ to 0, we~have $\lambda_l=  \frac{\sqrt{p_l} |\mathcal{K}_l|}{2p_l |\mathcal{K}_l| + \sigma_l^2}$ and $\mu_l=\frac{K}{2}$, which completes the proof.

\section{Proof of Theorem \ref{theo2}} \label{appc}
The solution $\mathcal{K}_l$ obtained by Algorithm \ref{alg0} satisfies the condition that for all $i\in \mathcal{K}_l$ and $j\notin\mathcal{K}_l$, it holds that $|h_{i,l}|\geq |h_{j,l}|$. Next, we need to demonstrate that the optimal $\mathcal{K}_l$ also adheres to this condition and that Algorithm~\ref{alg0} exhaustively explores all such sets. To begin with, we prove the first statement by contradiction. Suppose that in the optimal $\mathcal{K}_l$, device $i$ has the minimum channel gain, and there exists $j\notin \mathcal{K}_l$ with $|h_{j,l}|>|h_{i,l}|$. Replacing device $i$ with $j$ yields a new set $\mathcal{K}_l^\prime$ with $|\mathcal{K}_l|=|\mathcal{K}_l^\prime|$ and the associated optimal scaling factor is  $p_l^\prime=\min_{k\in \mathcal{K}_i^\prime}\{ |h_{k,l}|^2 \}> |h_{i,l}|$, thus enabling a larger scaling factor and a lower MSE. This contradicts the optimality of the original $\mathcal{K}_l$, confirming the first claim. Then, Algorithm \ref{alg0} constructs candidate sets in descending order of~channel gain, thereby enumerating all sets that satisfy the required condition. Together, these results complete the proof.

\end{document}